\documentclass[conference]{IEEEtran}

\usepackage{ifpdf}

\usepackage{cite}

\usepackage{soul}

\ifCLASSINFOpdf
 
\else
  
\fi

\usepackage[cmex10]{amsmath}

\usepackage{algorithmic}

\usepackage{array}

\usepackage{url}

\usepackage{graphicx}
\usepackage{amssymb}
\usepackage{bbold}
\usepackage{amsthm}
\usepackage{xcolor}

\newtheorem{theorem}{Theorem}

\newtheorem{proposition}{Proposition}
\newtheorem{example}{Example}
\newtheorem{definition}{Definition}

\newcommand{\IC}[2]{\ensuremath{IC_{\scriptscriptstyle #1}\left(#2\right)}}

\newcommand{\RT}{\ensuremath{\mathfrak{T}}}
\newcommand{\Rtens}[2]{\ensuremath{\RT({#1};{#2})}}
\newcommand{\CIwyn}{\ensuremath{CI_\text{\sf Wyn}}}
\newcommand{\CIgk}{\ensuremath{CI_\text{GK}}}
\newcommand{\Twyn}{\ensuremath{T_\text{\sf Wyn}}}
\newcommand{\Tgk}{\ensuremath{T_\text{GK}}}

\makeatletter
\def\footnoterule{\relax%
  \kern-5pt
  \hbox to \columnwidth{\hfill\vrule width 1.0\columnwidth height 0.4pt\hfill}
  \kern2pt}
\makeatother

\begin{document}

\title{Lower Bounds for Interactive Function Computation via Wyner Common Information}

\author{

\IEEEauthorblockN{Shijin Rajakrishnan\textsuperscript{1}}
\IEEEauthorblockA{IIT Madras, Chennai\\
Email: ee12b128@ee.iitm.ac.in}

\and
\IEEEauthorblockN{Sundara Rajan S\textsuperscript{1}}
\IEEEauthorblockA{IIT Madras, Chennai\\
Email: ee11b130@ee.iitm.ac.in}

\and
\IEEEauthorblockN{Vinod Prabhakaran}
\IEEEauthorblockA{
TIFR, Mumbai\\
Email: vinodmp@tifr.res.in}
}
\maketitle

\begin{abstract}
The question of how much communication is required between collaborating parties to compute a function of their data is of fundamental importance in the fields of theoretical computer science and information theory. In this work, the focus is on coming up with lower bounds on this. The information cost of a protocol is the amount of information the protocol reveals to Alice and Bob about each others inputs, and the information complexity of a function is the infimum of information costs over all valid protocols. For the amortized case, it is known that the optimal rate for the computation is equal to the information complexity. Exactly computing this information complexity is not straight forward however. In this work we lower bound information complexity for independent inputs in terms of the Wyner common information of a certain pair of random variables. We show a structural property for the optimal auxiliary random variable of Wyner common information and exploit this to exactly compute the Wyner common information in certain cases. The lower bound  obtained through this technique is shown to be tight for a non-trivial example - equality (EQ) for the ternary alphabet. We also give an example to show that the lower bound may, in general, not be tight.

\end{abstract}

\IEEEpeerreviewmaketitle
\footnotetext[1]{{Authors contributed equally.}}
\section{Introduction}
The amount of communication required by two parties to compute a function of their data is a central question in theoretical computer science and also information theory. Since the seminal work of Yao~\cite{Yao79}, much progress has been made on understanding {\em communication complexity} in computer science literature. While early progress was based on combinatorial techniques~\cite{KushilevitzNi97}, more recently advances in the area have centered around the notion of information
complexity, which measures the {\em amount of information} learned by the parties about each other's inputs from a protocol's transcript, rather than the {\em number of bits} in a
protocol's transcript, if it should compute a function (somewhat) correctly.
Specifically, if the inputs $X,Y$ of the parties come from a distribution $\mu$, the {\em information cost} of a protocol (for computing) $\Pi$ whose transcript is denoted by $M$ is defined as
\[  I(X;M|Y) + I(Y;M|X).\]
{\em Information complexity} is the infimum of information costs of valid protocols, i.e., protocols which allow the parties to compute within the desired error performance, and is denoted by $IC_{XY}(Z)$ for the computation of a function $Z=f(X,Y)$.
 
This quantity has a close connection to the problem of interactive source coding and interactive function computation studied in information theory literature. In particular, works by Kaspi~\cite{Kaspi85} and Ma and Ishwar~\cite{MaIs13} show that information complexity for zero-error is precisely the rate of communication required to compute with asymptotically vanishing error when the parties are allowed to code over long blocks of independent, identically distributed inputs. While, in general, computing information complexity is not straightforward, it is known exactly for some interesting examples~\cite{MaIs13} and an algorithm, albeit with run-time exponential in the alphabet size, for approximating it has been proposed~\cite{BravermanSc15}.

In~\cite{PrabhakaranPr14arxiv}, with the goal of better understanding information complexity, a monotonicity property of interactive protocols was leveraged to obtain lower bounds on the information complexity. The monotonicity property is that of the ``tension region'' of the views of the two users. Tension region of a pair of random variables was introduced in~\cite{PrabhakaranPr14} as a measure of dependence which cannot be captured using a common random variable. The question of how
well correlation {\em can } be captured by a random variable may be formulated
in terms of ``common information.'' Two different notions of common
information were developed in the 70's, $\CIgk(A;B)$ by G\'acs-K\"orner
\cite{GacsKo73}, and $\CIwyn(A;B)$ by Wyner \cite{Wyner75}.
\begin{align}
\CIgk(A;B)&=\max_{\substack{p_{Q|A,B}:\\Q-A-B\\Q-B-A}} I(Q;A,B)\\
\CIwyn(A;B)&=\min_{\substack{p_{Q|A,B}:\\A-Q-B}} I(Q;A,B) \label{eq:WynerCI}
\end{align} One can define
corresponding notions of tension as the gap between mutual information
(which accounts for all the correlation, but may not correspond to a common
random variable) and common information.  More precisely, one can define the
non-negative tension quantities $\Tgk(A;B)=I(A;B)-\CIgk(A;B)$ and
$\Twyn (A;B)=\CIwyn(A;B)-I(A;B)$.  These notions of tension were identified in
\cite{PrabhakaranPr14} as special cases of a unified 3-dimensional notion of {\em
tension region}. 
 
The tension region of a pair of random variables was defined in~\cite{PrabhakaranPr14} as the following upward closed region.
\begin{definition}
For a pair of random variables $A,B$, their {\em tension region}  \Rtens AB is defined as 
\begin{align*}
&\Rtens AB = \{ (r_1,r_2,r_3):\; \exists Q \text{ jointly distr. with } A,B \\
          &\text{ s.t. } I(B;Q|A) \le r_1, I(A;Q|B) \le r_2, I(A;B|Q) \le r_3 \}.
\end{align*}

\end{definition}
As shown in~\cite{PrabhakaranPr14}, without loss of generality, we may assume a cardinality bound $|\mathcal{Q}|\leq|\mathcal{A}||\mathcal{B}|+2$ on the alphabet $\mathcal{Q}$ in the above definition, where $\mathcal{A}$ and $\mathcal{B}$ are the alphabets of $A$ and $B$, respectively.

In \cite{PrabhakaranPr14}, an operational meaning was also obtained
for tension region in terms of a generalization of the common information problem of G\'acs and K\"orner. Tension region has proved useful in deriving converse results for secure computation. Specifically, it was used to strictly improve upon an upper bound of Ahlswede and Csisz\'ar~\cite{AhlswedeCs13} on the oblivious transfer capacity of channels~\cite{RaoPr14}.

Suppose $X$,$Y$ are the inputs and $A$,$B$ the outputs of the parties under a protocol. Let $M$ denote the transcript of the protocol. Let $V_A=(X,A,M)$ and $V_B=(Y,B,M)$ denote the views of the parties at the end of the protocol. The key monotonicity property we use is:
\begin{proposition}[Theorem~5.4 of \cite{PrabhakaranPr14}]\label{prop:monotonicity}
\[\Rtens {V_A}{V_B} \supseteq \Rtens XY.\]
\end{proposition}
A consequence of this is the following result:
\begin{theorem}
\label{thm:tens-ic}
For all $X,Y,Z$, 
\begin{align*}
 \IC{XY}{Z} &\ge \Twyn(XZ;YZ) - \Twyn(X;Y) \\&\qquad\qquad + I(X;Z|Y) + I(Y;Z|X).
\end{align*}

\end{theorem}
See~\cite{PrabhakaranPr14arxiv} for a more general result which implies the above lower bound. For the case of independent inputs the $\Twyn(X;Y)$ term goes to zero. We will give a proof of Theorem 1 for the case of independent inputs in Appendix A. While, the above bound is not always tight\footnote[2]{An example where this bound turns out not to be tight is that of computing the AND of two independent uniform bits $X,Y$, for which information complexity is known~\cite{MaIs13}.}, we present a non-trivial example where the bound turns out to give a tight result. It is worth noting that the technique of~\cite{MaIs13} does not easily yield this result.
\begin{example}
[ternary EQ]
Let $X,Y$ be independent and uniformly distributed over $\{0,1,2\}$. The goal is to compute the indicator for the event $(X=Y)$. Theorem~\ref{thm:tens-ic} gives a lower bound of $H_2\left(\frac{2}{3}\right)+\log_2(3)$ which can be shown to be tight.
\end{example}

{The equality (EQ) function, which determines whether two parties have the same inputs, has been studied extensively. To the best of our knowledge, the only lower bound on information complexity available is the trivial $IC_{XY}(Z)\geq I(X;Z|Y)+I(Y;Z|X)$. The best available upper bound is $4.5$ for $k$-ary EQ computation, for any probability distribution over the inputs} \cite{braverman2012interactive}. In this paper, we obtain both lower bounds and upper bounds on the information complexity of the EQ function for uniformly distributed inputs. 
To evaluate our lower bound of Theorem \ref{thm:tens-ic}, we need to compute Wyner common information  (or an equivalent quantity given in (\ref{eq:ICcomp})). Note that computing Wyner common information is, in general, not straightforward~\cite{Witsenhausen76}. Using standard techniques based on Carath\'eodory's theorem, an upper bound of $|{\mathcal Q}| \leq |{\mathcal A}|\times|{\mathcal B}|+2$ on the auxiliary random variable $Q$ of \eqref{eq:WynerCI} is available. We show that it is enough to consider a potentially smaller cardinality for ${\mathcal Q}$ which depends on the number of maximal cliques of the bipartite characteristic graph of $p_{A,B}$ -- this is the bipartite graph on ${\mathcal A}\times{\mathcal B}$ such that there is an edge between $a\in{\mathcal A}$ and $b\in{\mathcal B}$ if $p_{A,B}(a,b)>0$ -- such that conditioned on each element of $q\in {\mathcal Q}$, the characteristic graph of $p_{A,B|Q=q}$ is a distinct clique (Theorem~\ref{th:main}). This then allows us to compute Wyner common information exactly for certain examples of interest (Section~\ref{sec:examples}). In particular, the resulting lower bound turns out to be tight for the ternary EQ example above. We also give a randomized protocol for the 4-ary EQ problem which performs better than deterministic protocols in terms of its information cost, but here our lower bound does not meet the upper bound given by the protocol.

\section{Problem Formulation}
Alice and Bob get inputs $X$ and $Y$ respectively from a joint distribution $p_{XY}(x,y)$, their common objective being the computation of a function $Z=f(X,Y)$. They are connected by a channel which makes no errors in transmission. The \textit{protocol} to compute the function proceeds in a sequential manner as follows: initially Alice (or Bob) sends a message on the link, say $M_1$. Bob waits for the message to reach him and then sends a message $M_2$ on the link. The procedure iterates long enough for Alice and Bob to compute the function $Z$.
\begin{figure}[ht]
\begin{center}
\includegraphics[scale=0.25]{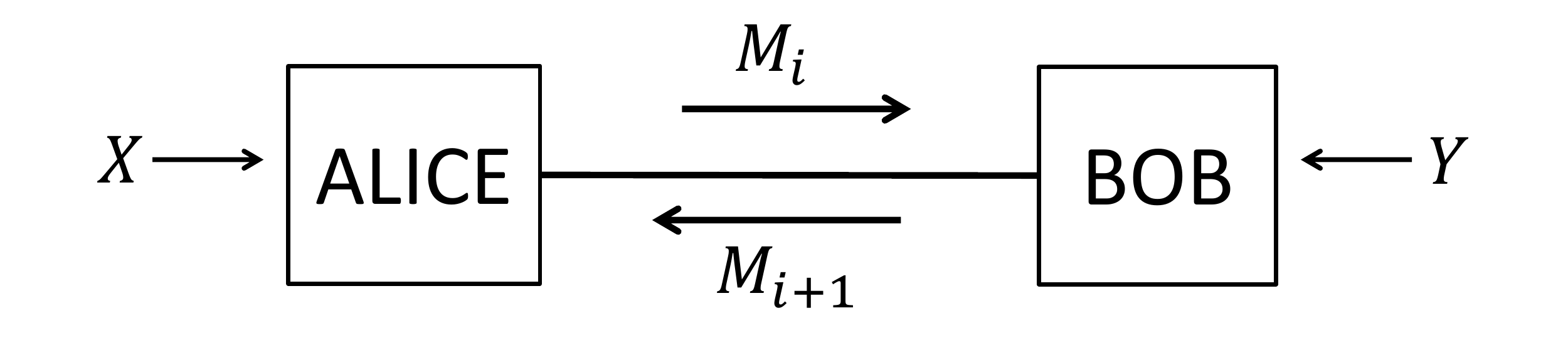}
\caption{The model for the two-party computation}
\label{fig:model}
\end{center}
\end{figure}
\\ Let $M^i=(M_1,M_2,...,M_i)$ denote the \textit{transcript} on the link till the $i^{th}$ stage and $M=(M_1,M_2,...)$ denote the final transcript at the end of the protocol. It is easy to see that the following two conditions are satisfied by any  protocol beginning at the Alice end: 
\begin{align}
&M_i - XM^{i-1} - Y && \forall \quad \mathrm{Odd} \quad i \label{eq:Alice-Markov}\\
&M_i - YM^{i-1} - X && \forall \quad \mathrm{Even}\quad  i \label{eq:Bob-Markov}
\end{align}

The entropy of the final transcript $H(M)$ is a lower bound for the average number of bits needed for the protocol. Further the \textit{information complexity} is a lower bound on $H(M)$. To prove this, we first prove the following inequality.\\ \indent Let $i$ be odd. Now,
\begin{align}
H(Y|M^{i-1})&\geq H(Y|M^i) \nonumber \\
 H(Y|M^{i-1})-H(Y|XM^{i-1})&\overset{(a)}{\geq} H(Y|M^i)-H(Y|XM^i)\nonumber \\
I(X;Y|M^{i-1})&\geq I(X;Y|M^{i}) 
\end{align}
where $(a)$ is due to $H(Y|XM^{i-1})=H(Y|XM^i)$, as $M_i - XM^{i-1} - Y$. The same inequality can be obtained for the case when $i$ is even, with a similar argument. A consequence of this is 
\begin{equation}
I(X;Y) \geq I(X;Y|M).
\label{eq:infoineq}
\end{equation}
\indent Now,

\begin{align}
H(M) &\geq I(M;XY)= I(X;M) + I(Y;M|X)\nonumber \\ &=I(X;YM)-I(X;Y|M)+I(Y;M|X)\nonumber\\&=I(X;Y)+I(X;M|Y)-I(X;Y|M)+I(Y;M|Y)\nonumber\\
&\overset{(a)}{\geq} I(X;M|Y) + I(Y;M|X) \geq IC_{XY}(Z)
\label{eq:HMGenLB}
\end{align}
where $(a)$ is due to (\ref{eq:infoineq}).\\
\indent Now, in the \textit{amortized} case, when we consider a block of independent identically distributed inputs of length $n$ and a sequence of schemes, one for each block length $n$, the following theorem, proved in \cite{MaIs13,el2011network}, gives the minimum rate of communication needed to compute a function with a vanishing probability of block error. The rate $R$ of a scheme is defined as the total number of bits exchanged divided by the block length. A rate $R$ is said to be {\em achievable} if there is a sequence of schemes whose probability of error goes to $0$ as $n \rightarrow \infty$. The optimal rate $R^*$ is the infimum of all achievable rates. 

\begin{theorem}
The optimal amortized rate $R^*$ for computing the function $Z=f(X,Y)$ is

\begin{equation}
R^* = \underset{M}{\inf}\left[I(X;M|Y) + I(Y;M|X)\right]=IC_{XY}(Z)
\label{eq:optim}
\end{equation}
where the infimum is over all $M=(M_1,M_2,...)$ satisfying the Markov chain conditions in (\ref{eq:Alice-Markov}) and (\ref{eq:Bob-Markov}),
and $H(Z|YM)=H(Z|XM)=0$.
\end{theorem}
\subsection{Lower bounding information complexity via Wyner common information}

Wyner tension, as defined in Section I can be written as:
\begin{equation}
 \Twyn(U;V)= \inf_{\substack{p_{Q|U,V}:\\U-Q-V}}[I(U;Q|V) + I(V;Q|U)]
 \label{eq: tension}
\end{equation}
Let $X$ and $Y$ be independent, from Theorem \ref{thm:tens-ic}, we can write
\begin{align*}
\IC{XY}{Z} &\ge \Twyn(XZ;YZ) + I(X;Z|Y) + I(Y;Z|X)
\end{align*}
Rewriting this is in a form suitable for our computation,
\begin{align}
\IC{XY}{Z} &\geq H(X|Y)+H(Y|X) \nonumber\\ & \quad - \sup_{\substack{p_{Q|U,V}:\\U-Q-V}} [H(U|Q)+H(V|Q)] \label{eq:ICcomp}
\end{align}
where $U=XZ$ and $V=YZ$. The problem now is to compute the supremum term in (\ref{eq:ICcomp}), where the auxiliary random variable $Q$ is such that given $Q$, the random variables $U$ and $V$ are independent. Given $Q=q$, for $U$ and $V$ to be conditionally independent, the edges in the characteristic graph should necessarily form a bipartite clique as shown in Fig \ref{fig:wyner1}. We first classify all the possible elements of $\mathcal{Q}$ into various classes, based on the characteristic graph formed by $U,V|Q=q$. We group all the elements with the same underlying bipartite clique into the same class. Now, since in a bipartite graph with a finite number of vertices in each vertex set, there are only finitely many bipartite cliques, we have a finite number of classes. Further, we combine several classes into one by looking only at maximal bipartite cliques, since a non maximal clique is just a special case of a maximal clique with some probability values being zero. Thus, the classes for a given $U,V$ distribution are those, each of which correspond to one maximal bipartite clique of the characteristic graph of $U,V$. Fig \ref{fig:Wyner-Tern} gives an example of such classes for a particular distribution. We can narrow down the search space of the alphabet $\mathcal{Q}_{opt}$ of an optimal auxiliary r.v ${Q}_{opt}$, which leads to the maximum value of $H(U|Q) + H(V|Q)$, with Theorem 3.

\begin{theorem} \label{th:main}

For a given $U,V$, to find the corresponding $\mathcal{Q}_{opt}$, it is sufficient to consider alphabets $\mathcal{Q}$ such that no two elements of $\mathcal{Q}$ are from the same class.

\end{theorem}
\begin{proof}
Consider Figure \ref{fig:wyner1}, which is a maximal bipartite clique of the characteristic graph of $U,V$, with left-degree $k$ and right-degree $l$. Now assume there are two elements of $\mathcal{Q}$, namely $q_0$ and $q'_0$ from the same class $q_1$, as shown in Figure \ref{fig:wyner2}. Each probability term $p_i$ refers to some $p_{QUV}(q,u,v)$, for example, in Figure \ref{fig:wyner2}, $p_2$ is the probability $p_{QUV}(q_0,1,2)$.

\begin{figure}[ht]
\begin{center}
\includegraphics[scale=0.25]{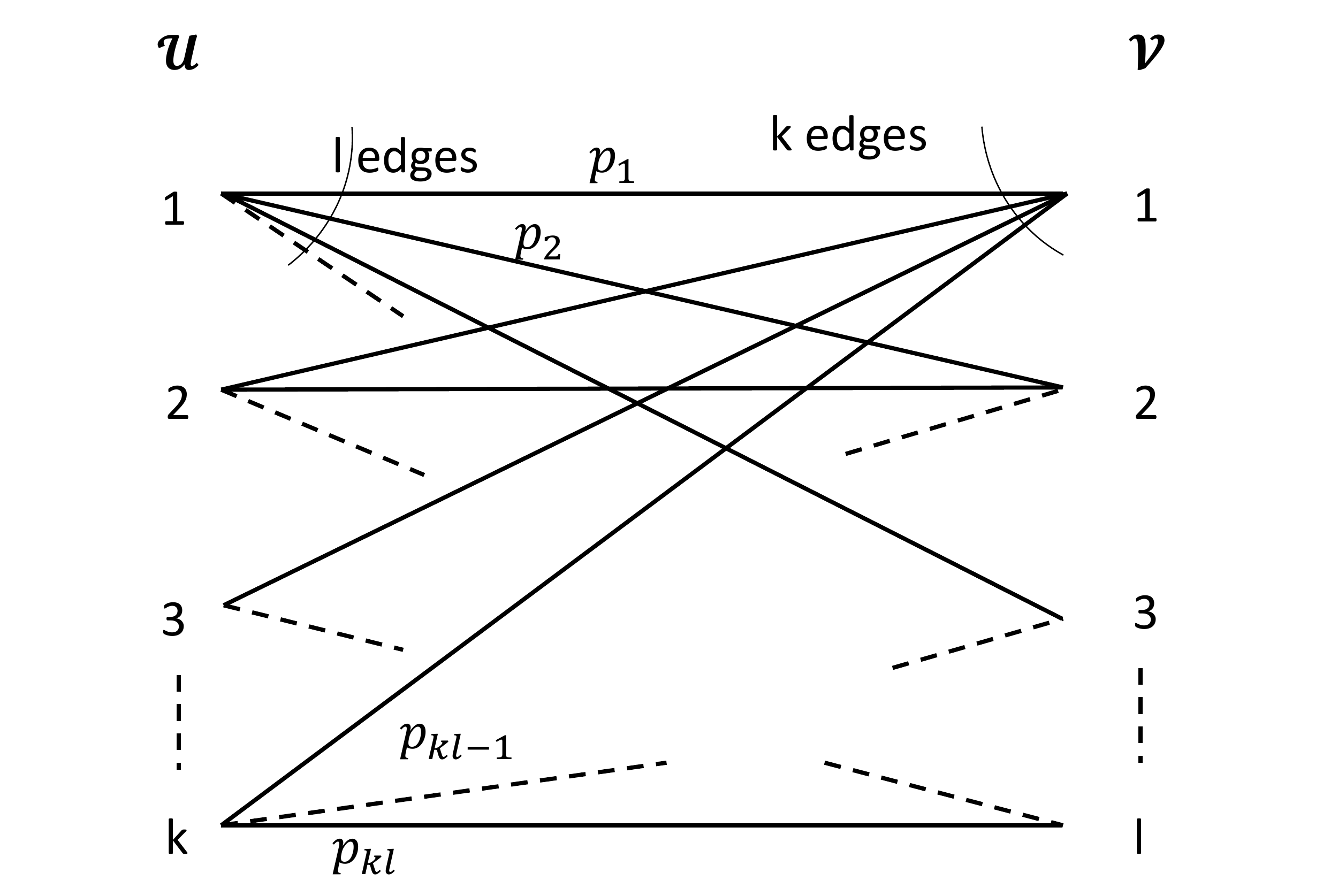}
\caption{Characteristic graph of $U,V|Q=q$ when $q$ is of class $q_1$}
\label{fig:wyner1}
\end{center}
\end{figure}
\begin{figure*}[ht]
\begin{center}
\includegraphics[scale=0.6]{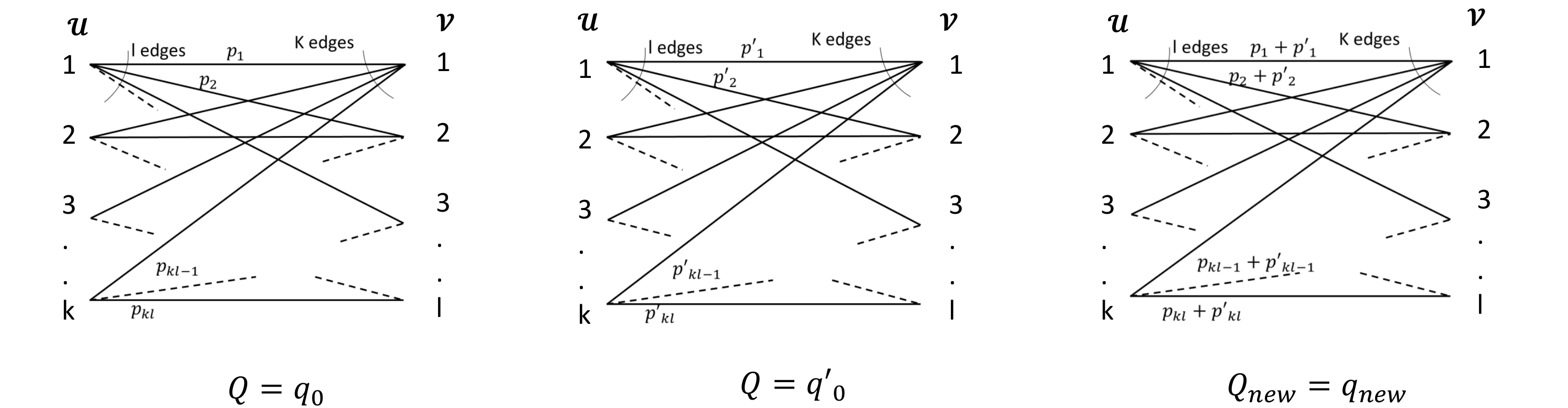}
\caption{If two elements of $\mathcal{Q}$ are from the same class, they can be merged to form a new $\mathcal{Q}_{new}$.}
\label{fig:wyner2}
\end{center}
\end{figure*}
For all random variables $Q$, whose alphabets have two elements from the same class, we can construct a ${Q}_{new}$ such that $H(U|Q) + H(V|Q) \leq H(U|Q_{new}) + H(V|Q_{new})$, by adding the weights of, and merging the corresponding edges of $q_0$ and $q_0'$, with all other elements remaining unchanged as shown in Figure \ref{fig:wyner2}. To prove this, we first prove that $H(U|Q)\leq H(U|Q_{new})$.\\
\begin{align}
H(U|Q) &= \sum_q p_Q(q)H(U|Q=q) \nonumber \\
&= \sum_q p_Q(q)H_{k}\Bigg(\frac{p_1+p_2+..+p_l}{p_Q(q)},\nonumber\\&\frac{p_{l+1}+..+p_{2l}}{p_Q(q)},...,\frac{p_{(k-1)l+1}+..+p_{kl}}{p_Q(q)} \Bigg)
\end{align}
\indent We now need to prove that
\begin{align}
p_Q(q_0)&H(U|Q=q_0)+p_Q(q_0')H(U|Q=q_0')\nonumber\\&\leq p_{Q_{new}}(q_{new})H(U|Q_{new}=q_{new}) \label{eq:tpt}
\end{align}
since the other terms in the summation are same for ${Q}$ and $Q_{new}$.
It is easy to see that $p_Q(q_0)=\sum_{i=1}^{kl} p_i$, and $p_Q(q_0')=\sum_{i=1}^{kl} p_i'$. Therefore $p_{Q_{new}}(q_{new})=\sum_{i=1}^{kl} (p_i+p_i')=p_Q(q_0)+p_Q(q_0')$. Hence, (\ref{eq:tpt}) is equivalent to
\begin{align}
(p_0&+p_0')H\left(\frac{a_1+a_1'}{p_0+p_0'},\frac{a_2+a_2'}{p_0+p_0'},...,\frac{a_k+a_k'}{p_0+p_0'} \right) \nonumber\\&\geq 
 p_0H\left(\frac{a_1}{p_0},\frac{a_2}{p_0},...,\frac{a_k}{p_0} \right) + p_0'H\left(\frac{a_1'}{p_0'},\frac{a_2'}{p_0'},...,\frac{a_k'}{p_0'} \right) \label{eq:tpt1}
\end{align}
where $a_i=p_{(i-1)l+1}+p_{(i-1)l+2}+...+p_{il}$, likewise $a_i'=p_{(i-1)l+1}'+p_{(i-1)l+2}'+...+p_{il}'$, and $p_0=p_Q(q_0)$, $p_0'=p_Q(q_0')$.\\
From the \textit{Log-Sum inequality}, $(a_i+a_i')\log {\frac{a_i+a_i'}{p_0+p_0'}}\leq a_i \log {\frac{a_i}{p_0}} + a_i ' \log {\frac{a_i'}{p_0'}}$, and so
\begin{align}
\sum_{i=1}^k \left[(a_i+a_i')\log {\frac{a_i+a_i'}{p_0+p_0'}}\right] & \leq \sum_{i=1}^k \left[ a_i \log {\frac{a_i}{p_0}} + a_i ' \log {\frac{a_i'}{p_0'}} \right] \nonumber 
\end{align}
and (\ref{eq:tpt1}) follows.\\
\indent Thus $H(U|Q)\leq H(U|Q_{new})$, and $H(V|Q)\leq H(V|Q_{new})$ can be proved with an equivalent argument.
\end{proof}

\section{{Lower Bounds on Information Complexity of EQ via Wyner Common Information}}\label{sec:examples}
We restrict our attention to inputs $X$ and $Y$ which are independent and uniformly distributed.

\subsection{Ternary EQ computation} 

Alice's and Bob's inputs, $X$ and $Y$ are independent and come uniformly from a distribution over ternary alphabets, say ${1,2,3}$. The function they want to compute is $Z=\mathbb{1}_{X=Y}$, the EQ function for a ternary alphabet. 

\begin{figure*}[ht]
 \begin{center}
  \includegraphics[scale = 0.6]{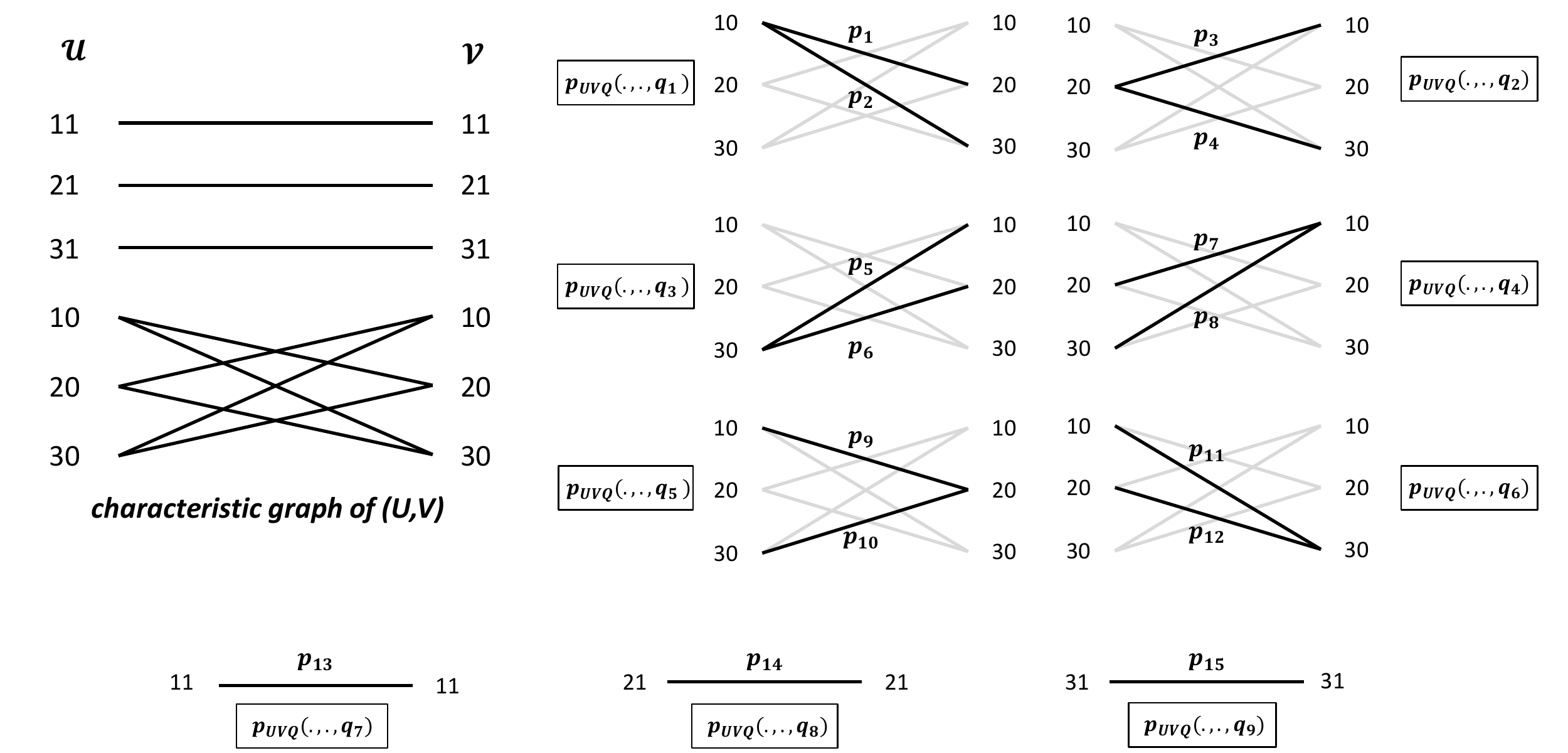} 
  \caption{Characteristic graph and Q classes for ternary EQ computation}
  \label{fig:Wyner-Tern}
 \end{center}
\end{figure*}

From Theorem \ref{th:main}, we can restrict the cardinality of $\mathcal{Q}$ to 9, where the different classes are shown in Figure \ref{fig:Wyner-Tern}.
From the uniform input distribution, we have $p_{XZ,YZ} = \frac{1}{9}$, and this leads to the constraints $\sum_q p_{XZ,YZ,Q}(u,v,q)=\frac{1}{9} \forall (u,v)$.\\
\begin{align}
p_i &= \frac{1}{9} \quad \forall i \in [13,15] \nonumber \label{eq:TernaryEQ}\\
p_{1} + p_{9} = \frac{1}{9}; \quad
p_{2} &+ p_{7} = \frac{1}{9}; \quad
p_{3} + p_{11} = \frac{1}{9};\nonumber  \\
p_{4} + p_{8} = \frac{1}{9}; \quad
p_{5} &+ p_{12} = \frac{1}{9}; \quad
p_{6} + p_{10} = \frac{1}{9}; 
\end{align}

\begin{align}
\text{Now, }\quad H(U|Q) + H(V|Q) &= \sum_{i=1}^{6}p_Q(q_i)H_2\left[\frac{p_{2i-1}}{p_Q(q_i)}\right] \nonumber \\
 &\leq \sum_{i=1}^{6}p_Q(q_i)
 =\frac{2}{3}
\end{align}
where we have used the fact that $H_2(\cdot)\leq 1$ and the set of equations in (\ref{eq:TernaryEQ}). So now from (\ref{eq:ICcomp}), we get $IC_{XY}(Z) \geq H(X)+H(Y)-\frac{2}{3} = 2\log 3 - \frac{2}{3} = 2.5033$. 

Consider the following protocol for the upper bound; in the \textit{amortized} case, this has to repeated over the block of inputs.\\
\noindent\makebox[\linewidth]{\rule{9cm}{0.5pt}}\label{Protocol:log3bitEQ}  
\textbf{Protocol 1: Ternary EQ computation}\\
1. Alice sends a symbol from a ternary alphabet indicating her input to Bob.\\
2. Bob locally computes $Z=\mathbb 1_{X=Y}$, and sends the resultant bit to Alice.\\
\noindent\makebox[\linewidth]{\rule{9cm}{0.5pt}}
\noindent The information cost for the above protocol is $H(X)+H(Z)=\log(3)+H_2(\frac{1}{3}) = 2.5033$. 
Thus we see that the lower bound developed is tight in this example. This protocol could be represented as ${Q}_{opt}$ in Figure \ref{fig:Wyner-Tern} as follows: $p_i=\frac{1}{9}\forall i \in [1,6],
p_i=0 \forall i \in [7,12],
p_i=\frac{1}{9} \forall i \in [13,15]$.

\subsection{Two bit EQ computation}
\label{2bitEQ_prob}
Alice and Bob communicate in order to compute the EQ function for two bits, $Z=\mathbb 1_{(X_0,X_1)=(Y_0,Y_1)}$, where all the bits are i.i.d $\mathcal{B}$($\frac{1}{2}$). 

From Theorem \ref{th:main}, it is sufficient if we look at $\mathcal{Q}$ S.T $|\mathcal{Q}|\leq 18$. The 18 classes in this case consists of 2 types of maximal bipartite cliques; one with 3 edges and the other with 4 edges each. A similar analysis along the lines elucidated in the case of ternary EQ would result in $H(U|Q)+H(V|Q) \leq 1.5$, and the upper bound is attained when the distribution on the 4 edge classes is uniform, i.e. the probability metric associated with each edge of a 4-edge class is same and equal to $\frac{1}{32}$. This implies that $\sup_{\substack{p_{Q|U,V}\\U-Q-V}}[H(U|Q)+H(V|Q)] = 1.5$, and hence $IC_{XY}(Z)\geq 2.5$.

We derive an upper bound on $IC_{XY}(Z)$ by giving a randomized protocol.\\
\noindent\makebox[\linewidth]{\rule{9cm}{0.5pt}}\label{Protocol:2bitEQ_rand} \\
\textit{Definitions:} Let Alice's input $X$ be uniform in $\mathcal{A}=\{1,2,3,4\}$, and Bob's input $Y$ uniform in $\mathcal{B}=\{1,2,3,4\}$. Define the sets $\mathbf{a}=\{1,2\},\mathbf{b}=\{1,3\},\mathbf{c}=\{1,4\},\mathbf{d}=\{2,3\},\mathbf{e}=\{2,4\},\mathbf{f}=\{3,4\}$.\\\\
\textbf{Protocol 2: Two bit EQ computation - Randomized}\\
{ 1. Alice uniformly picks $\mathbf{u}\in\{$a,b,c,d,e,f$\}$ such that $X\in \mathbf{u}$, and sends it to Bob.\\
2. If $Y\in \mathbf{u}$, Bob sends 1. Else he sends 1 or 0 with equal probability. \\\indent If Bob's message is 0, the protocol terminates and $Z=0$.\\\indent If it is 1, protocol proceeds to step 3.\\
3. Alice reveals her input.\\
4. Bob computes $Z$ and sends the result to Alice.\\}
\noindent\makebox[\linewidth]{\rule{9cm}{0.5pt}} \\
 \indent 
 If $X=Y$, which occurs with probability $\frac 1 4$, both parties learn 2 bits.
 If $X \neq Y$, but $Y \in \mathbf{u}$, which happens with probability $\frac{1}{4}$, then Bob sends $1$, and thus they proceed to step 3. If $Y \notin \mathbf{u}$, then Bob sends $1$ with probability $\frac{1}{2}$. So, given that Bob sends 1, Bob's input $Y \in \mathbf{u}$ with probability $\frac{1}{2}$. Hence if the protocol goes to step 3, Alice's uncertainty about Bob's input is $H_3(\frac{1}{2},\frac{1}{4},\frac{1}{4})=1.5$ at the end of the protocol. If it stops at step 2, Alice and Bob each would have learnt only 1 bit about each other. Therefore, the information cost is,
$\frac{1}{4}(4)+\frac{1}{4}(4-\frac{3}{2})+\frac{1}{4}(4-\frac{3}{2})+\frac{1}{4}(2) = 2.75$

\subsection{Wyner Tension for EQ with an arbitrary sized input alphabet}
{$X$ and $Y$ are uniformly and independently distributed from an alphabet} of cardinality $k$, and they want to compute $Z$, the EQ function which takes the value 1 when their inputs are equal. Now the maximal bipartite cliques in this new setting for the characteristic graph of $U=XZ$ and $V=YZ$ will be functions of $k$. 
For each $i$ {there will} be $^kC_i$ maximal bipartite cliques with $i$ nodes from the $U$ side, with $Z=0$ and $k-i$ nodes from the $V$ side, with $Z=0$. So the total number of classes would be $\sum_{i=1}^{k-1}$ $^kC_i + k$, where the final $k$ classes are for the $Z=1$ case, each containing one edge. For the $Z=0$ cliques, we refer to a maximal bipartite clique with $i$ nodes from the $U$ set as belonging to a class {$\mathcal{L}_i$}. Given some edge with $Z=0$, connecting $U=u$ and $V=v$: $(u,v)$, we can enumerate the number of {classes, $\mathcal{L}_i$} that contains the edge. Each edge $(u,v)$, with $u \neq v$, occurs in classes {$\mathcal{L}_1$} only once, in classes  {$\mathcal{L}_2$} $^{k-2}C_1$ times, and in general, occurs $^{k-2}C_{i-1}$ times in the classes {$\mathcal{L}_i$}.
\indent Now as in the earlier cases, each of the edges has a probability $p_{U,V,Q}(u,v,q)$ associated with it, which leads to a set of constraints: 
\begin{equation}
\label{eq:genCons}
 p_{U,V}(u,v) =\sum_{q}p_{U,V,Q}(u,v,q)= \frac{1}{k^2}\quad \forall (u,v) 
\end{equation}
In addition to these constraints the $p_{U,V,Q}(u,v,q)$ should be such that $U-Q-V$. Now,
\begin{align}
 &H(U|Q) + H(V|Q) \nonumber\\&= \sum_{q_{i}}\left[H(U|Q=q_i)p_Q(q_i) + H(V|Q=q_i)p_Q(q_i)\right] \nonumber\\ 
 &\leq \sum_{\mathcal{L}_1}p_Q(q_i)\log(k-1) + \sum_{\mathcal{L}_2}p_Q(q_i)(1 + \log(k-2)) \nonumber \\ 
& \quad + \dots + \sum_{\mathcal{L}_{k-1}}p_Q(q_i)\log(k-1) \label{eq:ICgen}
\end{align}

\textit{Case I: k is even}: Using the fact that if we have 2 non-negative integers $a$ and $b$ such that $a + b = k$ (a constant), the maximum value of $ab$ is when $a=b=\frac{k}{2}$, we get $(\log(i)+\log(k-i)) \leq (\log(\frac{k}{2}) + \log(\frac{k}{2}))$. Using this in (\ref{eq:ICgen}), we get
\begin{align}
 H&(U|Q) + H(V|Q) \nonumber\\&\leq \sum_{\mathcal{L}_1}p_Q(q_i)2\log\frac{k}{2} 
 +..+ \sum_{\mathcal{L}_{k-1}}p_Q(q_i)2\log\frac{k}{2} \nonumber 
 \end{align}
 \begin{align}
 =2\log(\frac{k}{2})\sum_{\mathcal{L}_1,..,\mathcal{L}_{k-1}}p_Q(q_i) =2(1-\frac{1}{k})\log(\frac{k}{2}) \label{eq:evenKTension} 
 \end{align}

Consider the distribution $p(u,v,q) = \frac{1}{k^2\left(^{k-2}C_{\frac{k-2}{2}}\right)}$ for all the edges in classes {$\mathcal{L}_{\frac{k}{2}}$}, and $p=0$ for all the other edges in the $Z=0$ set( Of course, for all the edges with $Z=1$, we need $p=\frac{1}{k^2}$ so as to satisfy the constraints in (\ref{eq:genCons})). It is easy to verify that this distribution ensures that $U-Q-V$, and hence is a valid $Q$ choice. For this distribution, the value of $H(U|Q) + H(V|Q)$ is,
\begin{equation*}
 ^kC_{\frac{k}{2}} \cdot \frac{\left(\frac{k}{2}\right)\left(\frac{k}{2}\right)}{k^2\left(^{k-2}C_{\frac{k-2}{2}}\right)}\cdot2\log\frac{k}{2} =\quad 2\left(1-\frac{1}{k}\right)\log\frac{k}{2}
\end{equation*}

and so, $\sup_{\substack{p_{Q|U,V}:\\U-Q-V}}H(U|Q) + H(V|Q) = 2\left(1-\frac{1}{k}\right)\log\frac{k}{2}$.\\
From (\ref{eq:ICcomp}), we get \\$IC_{XY}(Z) \geq 2\log(k) - 2\left(1-\frac{1}{k}\right)\log\frac{k}{2} = 2 + \frac{2}{k}\log\frac{k}{2}$.
\\
\textit{Case II: k is odd}: Like in the previous case, one can see that
\begin{equation*}
 H(U|Q) + H(V|Q)\leq \left[\log(\frac{k-1}{2}) + \log(\frac{k+1}{2})\right]\left(1-\frac{1}{k}\right)
\end{equation*}
Again, we can consider the distribution $p = \frac{1}{k^2\left(^{k-2}C_{\frac{k-1}{2}}\right)}$ for all the edges in classes {$\mathcal{L}_{\frac{k+1}{2}}$}, so that
\begin{align*}
 H(&U|Q) + H(V|Q) = \left[\log\left(\frac{k-1}{2}\cdot\frac{k+1}{2}\right)\right]\left(1-\frac{1}{k}\right)
\end{align*}
So
$\sup_{\substack{p_{Q|U,V}:\\U-Q-V}}[H(U|Q) + H(V|Q)] = \left[\log\left(\frac{k^2-1}{4}\right)\right]\left(1-\frac{1}{k}\right)$, and from  (\ref{eq:ICcomp}), we get \\
\\$IC_{XY}(Z) \geq 2\log(k) - \left[\log\left(\frac{k-1}{2}\cdot\frac{k+1}{2}\right)\right]\left(1-\frac{1}{k}\right)$.

\section{Conclusion}
{
In this paper we demonstrated a method for obtaining lower bounds on information complexity of functions under independent input distributions via computing Wyner common information. We showed the tightness of our lower bound for the ternary EQ function. For the 2-bit EQ function, our lower bound works out to 2.5, while we obtained an upper bound of 2.75 by giving a randomized protocol. 
  For the $k$-ary EQ function, our lower bound converges to 2 as $k \rightarrow \infty$. Repeated use of 2-bit EQ computation protocol gives an upper bound of 3.667 as $k \rightarrow \infty$.
}

\appendix[Proof of Theorem \ref{thm:tens-ic}]
\indent Consider the case when $X$ and $Y$ are independent. From (\ref{eq:infoineq}), $I(X;Y|M)\leq I(X;Y) =0$, and hence $I(X;Y|M)=0$. Using this, for any valid protocol with transcript $M$,
\begin{align*}
I(XZ;&YZ|M)=I(X;Y|M)+I(Z;Y|MX)\\& \quad \quad \quad  +I(X;Z|MY)+I(Z;Z|XYM)
\\&\leq 0 + H(Z|MX) +H(Z|MY) +H(Z|XYM)\\
&\overset{(a)}{=}0
\end{align*}
$(a)$ is because all the four terms are 0. Hence $I(XZ;YZ|M)=0$ and the Markov chain $XZ-M-YZ$.
\begin{align}
\text{Now, }I&(X;M|Y) + I(Y;M|X) \nonumber \\
&\overset{(a)}{=} I(X;MZ|Y) + I(Y;MZ|X) \nonumber \\ 
&\overset{(b)}{=} I(X;Z|Y) + I(Y;Z|X) \nonumber \\& \quad + I(XZ;M|YZ) + I(YZ;M|XZ)  \nonumber \\
&\overset{(c)}{\geq}I(X;Z|Y) + I(Y;Z|X) + \Twyn(XZ;YZ) \label{eq:HMTen}
\end{align}
\noindent where $(a)$ follows from the fact that $0 \leq I(X;Z|MY) \leq H(Z|MY) =0$, $(b)$ is true as $I(XZ;M|YZ)=I(X;M|YZ)+H(Z|MYZ)=I(X;M|YZ)$, $(c)$ is a result of the relaxation $XZ-M-YZ$. This implies that the information complexity of the setting $\IC{XY}{Z} \geq I(X;Z|Y) + I(Y;Z|X) + \Twyn(XZ;YZ)$, thus proving Theorem 1 for independent inputs.

\section*{Acknowledgments}

Vinod Prabhakaran would like to acknowledge useful discussions with Prakash Narayan and Shun Watanabe. The problem in Section \ref{2bitEQ_prob} is due to Shun Watanabe and was presented in~\cite{NarayanITW15}.

The authors would like to thank the Visiting Students' Research Programme (VSRP) of Tata Institute of Fundamental Research (TIFR), Mumbai for facilitating Shijin Rajakrishnan and Sundara Rajan's summer internship at TIFR. Vinod Prabhakaran's research was funded in part by a Ramanujan Fellowship from the Department of Science \& Technology, Government of India.

\end{document}